\documentclass[11pt]{article}
\usepackage[a4paper,margin=2.6cm]{geometry}
\usepackage{amsmath,amssymb,amsthm}
\usepackage{booktabs}
\usepackage{array}
\usepackage{enumitem}
\usepackage{xcolor}
\usepackage[colorlinks=true,linkcolor=blue!55!black,citecolor=blue!55!black,urlcolor=blue!55!black]{hyperref}

\theoremstyle{plain}
\newtheorem{theorem}{Theorem}[section]
\newtheorem{lemma}[theorem]{Lemma}

\newtheorem{proposition}[theorem]{Proposition}
\theoremstyle{definition}
\newtheorem{definition}[theorem]{Definition}
\newtheorem{example}[theorem]{Example}
\newtheorem{remark}[theorem]{Remark}

\newcommand{\N}{\mathcal N}
\newcommand{\Rch}{\mathcal R}

\title{A Reachable-State Operator Formulation of Deferred Acceptance:\\
Progress Invariants and Structural Diagnostics}
\author{Yoshiteru Ishida\\{\small ORCID: 0000-0003-1641-5385}}
\date{}

\begin{document}
\maketitle

\begin{abstract}
We give a reachable-state operator formulation of one-to-one deferred acceptance with
strict, possibly incomplete preference lists.  A local update operator is defined for each
active proposer, while a scheduler selects which local update is applied.  On the states
reachable from the canonical empty initial state, three invariants are immediate: the set
of proposed edges grows strictly, every proposer visits each acceptable receiver at most
once, and every receiver holds its most-preferred proposal received so far.  These
invariants yield termination after at most $|E|$ proposals and stability of every terminal
reachable state.  The classical rejection lemma then gives proposer optimality and
schedule independence.  We separate these trajectory statements from two logically
different results: the distributive-lattice structure of the full stable-matching set and
one-sided strategy-proofness.  A diagnostic table records which proof obligation is lost
when a model changes the bipartition, ordinal comparisons, proposal irreversibility, or
receiver choice rule.  The paper makes no new complexity claim; its purpose is a precise
operator-level account of the classical proof architecture and of the limits of that
architecture.
\end{abstract}



\section{Introduction}

Deferred acceptance (DA), introduced by Gale and Shapley~\cite{gale1962}, is usually
presented as a sequential algorithm: an unmatched proposer visits acceptable receivers in
preference order, and each receiver retains the best proposal received so far.  The
standard proof establishes termination, stability, and proposer optimality.  Closely
related but logically distinct results describe the lattice of all stable matchings and
the strategy-proofness of the proposer side~\cite{dubins1981,gusfield1989,knuth1997,roth1982}.

This note formalizes the sequential process as a family of local state updates.  A fixed
scheduler turns the family into a deterministic trajectory, but the proofs do not depend
on the scheduler.  The state space is restricted to states reachable from the canonical
empty initial state.  This restriction is essential: an arbitrary perfect matching can be
made a syntactic fixed point of a poorly chosen state encoding without being stable.

The contribution is organizational rather than complexity-theoretic.  We make explicit
which invariants prove which classical conclusions, and we distinguish the DA trajectory
from fixed-point formulations on lattices due to Adachi~\cite{adachi2000} and
Fleiner~\cite{fleiner2003}.  The resulting diagnostic language says only that a particular
classical proof no longer applies when an invariant is lost; it does not infer
intractability.

\section{Strict bipartite instances and reachable states}

Let $P$ be a finite set of proposers and $R$ a finite set of receivers.  Acceptability is
specified by a bipartite graph $G=(P\cup R,E)$.  Every vertex has a strict order over its
neighbors.  Matchings may leave vertices unmatched, and being unmatched is worse than any
acceptable partner.

For $p\in P$, write
\[
  L(p)=(r^p_1,r^p_2,\ldots,r^p_{d(p)})
\]
for the acceptable receivers in decreasing order of preference, where $d(p)=|\N(p)|$.

\begin{definition}[State]
A state is a pair $s=(\mu,q)$, where $\mu$ is a matching in $G$ and
\[
 q_p\in\{1,2,\ldots,d(p)+1\}\qquad(p\in P).
\]
The value $q_p$ is the position of the next receiver to whom $p$ will propose;
$q_p=d(p)+1$ means that $p$ has exhausted its list.  A proposer is \emph{active} in $s$
if it is unmatched under $\mu$ and $q_p\le d(p)$.  The active set is denoted $A(s)$.
\end{definition}

\begin{definition}[Canonical initial state]
The canonical initial state is $s_0=(\varnothing,q^0)$, where $q^0_p=1$ for every
$p\in P$.
\end{definition}

\begin{definition}[Local update]
Let $p\in A(s)$ and let $r=r^p_{q_p}$.  The local update $U_p(s)=(\mu',q')$ is obtained as
follows.
\begin{enumerate}[label=(\roman*)]
\item Set $q'_p=q_p+1$ and $q'_x=q_x$ for $x\ne p$.
\item If $r$ is unmatched, set $\mu'=\mu\cup\{pr\}$.
\item If $r$ is matched to $p'$ and $r$ prefers $p$ to $p'$, replace $p'r$ by $pr$.
\item Otherwise set $\mu'=\mu$.
\end{enumerate}
\end{definition}

A \emph{scheduler} chooses one element of $A(s)$ whenever $A(s)\ne\varnothing$.  For a
scheduler $\sigma$, define
\[
 \Phi_\sigma(s)=
 \begin{cases}
   U_{\sigma(s)}(s),&A(s)\ne\varnothing,\\
   s,&A(s)=\varnothing.
 \end{cases}
\]
Thus $\Phi_\sigma$ is a deterministic map once $\sigma$ is fixed.

\begin{definition}[Reachable state]
A state is reachable if it occurs on a trajectory
$s_0,s_1,s_2,\ldots$ with $s_{t+1}=\Phi_\sigma(s_t)$ for some scheduler $\sigma$.
The set of states reachable under at least one scheduler is denoted $\Rch$.
\end{definition}

\section{Trajectory invariants}

For a reachable state $s=(\mu,q)$, define the proposed-edge set
\[
 H(s)=\{\,p r^p_k: p\in P,\ 1\le k<q_p\,\}
\]
and the proposal potential
\[
 Q(s)=|H(s)|=\sum_{p\in P}(q_p-1).
\]

\begin{lemma}[Reachable-state invariants]\label{lem:invariants}
Along every trajectory from $s_0$ the following hold.
\begin{enumerate}[label=(\alph*)]
\item Each nonterminal update adds exactly one edge to $H(s)$, so $Q$ increases by one.
\item No acceptable edge is proposed more than once.
\item Every held edge belongs to $H(s)$.
\item Once a receiver becomes matched, it remains matched at every later state.
\item If receiver $r$ is matched in state $s$, then its partner is the most-preferred
proposer among all proposers from whom $r$ has received a proposal by state $s$.
\end{enumerate}
\end{lemma}

\begin{proof}
Parts (a)--(c) follow immediately from the pointer update and the canonical initial state.
Part (d) follows because an accepted proposal can only fill an unmatched receiver or replace
its current partner.  For (e), use induction on the trajectory length.  The assertion is
initially vacuous.  A receiver changes partner only when a newly proposing proposer is
preferred to its current partner, so after the update it holds the best proposal received so far.
\end{proof}

\begin{theorem}[Finite termination]\label{thm:termination}
For every scheduler, the trajectory from $s_0$ reaches a terminal reachable state after at
most $|E|$ proposals.
\end{theorem}

\begin{proof}
By Lemma~\ref{lem:invariants}(a), $Q$ increases by one at each nonterminal update.  By
Lemma~\ref{lem:invariants}(b), $Q\le |E|$.
\end{proof}

\section{Stability and proposer optimality}

A matching $\mu$ is stable if no acceptable edge $pr\notin\mu$ satisfies both: proposer
$p$ prefers $r$ to $\mu(p)$ (or is unmatched), and receiver $r$ prefers $p$ to $\mu(r)$
(or is unmatched).

\begin{theorem}[Terminal stability]\label{thm:stability}
The matching component of every terminal state reachable from $s_0$ is stable.
\end{theorem}

\begin{proof}
Let $s^*=(\mu^*,q^*)$ be terminal and suppose that $pr$ blocks $\mu^*$.  Because $p$
prefers $r$ to its terminal partner, or is unmatched, $p$ must have proposed to $r$ before
becoming matched to a less-preferred receiver or exhausting its list.  Hence $pr\in H(s^*)$.
Receiver $r$ did not retain $p$.  By Lemma~\ref{lem:invariants}(e), the terminal partner of
$r$ is preferred to $p$, contradicting that $pr$ blocks.  If $r$ is terminally unmatched,
then Lemma~\ref{lem:invariants}(d) implies that it received no proposal, also contradicting
$pr\in H(s^*)$.
\end{proof}

The next lemma is classical; it is stated to identify the additional ingredient needed for
optimality.

\begin{lemma}[Rejection lemma; Gale--Shapley]\label{lem:rejection}
If a receiver rejects a proposer during a DA trajectory from $s_0$, then that pair occurs
in no stable matching of the instance.
\end{lemma}

\begin{proof}
Assume for contradiction that some rejection of a pair $pr$ is the first rejection of an
edge belonging to a stable matching $\nu$.  At that moment $r$ holds a proposer $p'$ whom
$r$ prefers to $p$.  In $\nu$, proposer $p'$ is matched either to $r$ or to a receiver
$r'$ that $p'$ ranks below $r$: otherwise $p'$ must previously have been rejected by a
receiver it prefers to $r$, contradicting the choice of the first rejected edge of $\nu$.
If $\nu(p')=r$, then $\nu$ cannot also contain $pr$.  If $\nu(p')=r'\ne r$, the pair
$p'r$ blocks $\nu$, because $p'$ prefers $r$ to $r'$ and $r$ prefers $p'$ to $p=\nu(r)$.
Both cases are impossible.
\end{proof}

\begin{theorem}[Proposer optimality and schedule independence]\label{thm:optimality}
Every terminal trajectory from $s_0$ produces the proposer-optimal stable matching.
Consequently, its terminal matching is independent of the scheduler.
\end{theorem}

\begin{proof}
Let $\mu^*$ be a terminal DA matching.  Whenever proposer $p$ passes a receiver $r$, the
pair $pr$ has been rejected and therefore belongs to no stable matching by
Lemma~\ref{lem:rejection}.  Thus no stable matching can assign $p$ a receiver that $p$
prefers to $\mu^*(p)$.  By Theorem~\ref{thm:stability}, $\mu^*$ is itself stable, hence it
is proposer optimal.  The proposer-optimal stable matching is unique, so all schedulers
produce the same terminal matching.
\end{proof}

\begin{remark}[What is and is not a fixed point]
For a fixed scheduler, terminal reachable states are fixed points of $\Phi_\sigma$.  An
arbitrary syntactic state with no active proposer need not encode a valid DA history and
need not have a stable matching component.  Theorems~\ref{thm:stability} and
\ref{thm:optimality} therefore concern reachable terminal states, not all fixed points of
an unrestricted state map.
\end{remark}

\section{A small example}

\begin{example}
Let $P=\{m_1,m_2,m_3,m_4\}$ and $R=\{w_1,w_2,w_3,w_4\}$, with preferences
\[
\begin{array}{ll}
 m_1:w_4\succ w_1\succ w_3\succ w_2, & w_1:m_1\succ m_2\succ m_4\succ m_3,\\
 m_2:w_3\succ w_4\succ w_2\succ w_1, & w_2:m_2\succ m_4\succ m_3\succ m_1,\\
 m_3:w_2\succ w_3\succ w_4\succ w_1, & w_3:m_4\succ m_3\succ m_2\succ m_1,\\
 m_4:w_4\succ w_1\succ w_3\succ w_2, & w_4:m_1\succ m_2\succ m_3\succ m_4.
\end{array}
\]
The proposer-side DA outcome is
\[
 \{m_1w_4,m_2w_3,m_3w_2,m_4w_1\},
\]
and the receiver-side DA outcome is
\[
 \{m_1w_4,m_2w_2,m_3w_3,m_4w_1\}.
\]
Exhaustive enumeration shows that these are the only stable matchings, so their lattice is
a two-element chain.  This example illustrates extremal selection, but it is not used in
the proofs.
\end{example}

\section{Proof obligations and structural diagnostics}

The preceding results separate several proof obligations that are often grouped together.
The table records sufficient ingredients in the classical proof; it does not assert that
they are logically minimal or that their absence implies computational hardness.

\begin{table}[ht]
\centering
\small
\renewcommand{\arraystretch}{1.18}
\begin{tabular}{p{2.7cm}p{4.4cm}p{5.5cm}}
\toprule
Conclusion & Sufficient proof ingredient & What must be replaced when the ingredient is lost\\
\midrule
Termination & finite acceptable-edge set and no repeated proposals & a different potential,
cycle-elimination rule, or global termination argument\\
Stability & exhaustive proposal opportunity and receiver best-so-far invariant & a new dominance
or blocking-pair argument\\
Proposer optimality & rejection lemma & another certificate that rejected edges cannot belong to
a preferred stable outcome\\
Schedule independence & uniqueness of the proposer-optimal stable matching & a confluence or
canonical-selection theorem\\
Lattice extremality & lattice theorem for the complete stable set & a structural description of
all feasible stable outcomes\\
Strategy-proofness & mechanism-design argument in addition to DA correctness & a direct incentive
analysis\\
\bottomrule
\end{tabular}
\caption{Proof obligations in the classical deferred-acceptance theory.}
\label{tab:obligations}
\end{table}

\begin{proposition}[Diagnostic principle]\label{prop:diagnostic}
If a variant preserves the hypotheses used in Theorems~\ref{thm:termination} and
\ref{thm:stability}, then the same termination and stability proofs apply verbatim.  If a
hypothesis is not preserved, the corresponding conclusion is not disproved, but the
classical proof has an explicit missing obligation listed in Table~\ref{tab:obligations}.
\end{proposition}

\begin{proof}
The first statement follows by inspection of the two proofs.  The second is a logical
statement about proof reuse: removing an invoked hypothesis invalidates that derivation but
does not imply the negation of its conclusion.
\end{proof}

Some standard variants illustrate the distinction.
\begin{itemize}
\item Strict incomplete lists preserve the finite-edge, no-repeat, and best-so-far
invariants; the formulation above already includes this case.
\item Ties remove unique ordinal comparisons.  Weak, strong, and super stability require
different acceptance and blocking rules, so the relevant receiver invariant must be
restated~\cite{manlove2013}.
\item Stable roommates removes a fixed proposer--receiver bipartition.  The classical
proposer-optimal order is therefore unavailable, although Irving's separate algorithm
solves the strict complete roommates problem in polynomial time~\cite{irving1985}.
\item Constraints or matching with contracts may replace individual receiver choice by a
choice function.  Substitutability and irrelevance of rejected contracts are then the
appropriate structural conditions~\cite{hatfield2005}.
\end{itemize}

\section{Relation to lattice fixed-point formulations}

The state operator above is procedural.  It acts on reachable algorithm states and records
one DA trajectory.  By contrast, Adachi~\cite{adachi2000} and Fleiner~\cite{fleiner2003}
construct order-theoretic maps whose fixed points characterize stable outcomes, within the
tradition of Tarski's fixed-point theorem~\cite{tarski1955}.  These are different uses of
``fixed point.''  In particular, the present proofs use a strictly increasing finite
potential along a trajectory; they do not require global isotonicity of $\Phi_\sigma$ on a
complete lattice.

The lattice of stable matchings is likewise a theorem about the entire solution set, not a
consequence of the trajectory potential.  The proposer- and receiver-optimal DA outcomes
are its extremal elements~\cite{gusfield1989,knuth1997}.  One-sided strategy-proofness is a
further mechanism-design theorem~\cite{dubins1981,roth1982}; it is not proved by the
potential argument.

\section{Conclusion}

A reachable-state formulation prevents three common conflations.  First, a fixed point of a
syntactic state map is not automatically a stable matching; reachability from the canonical
initial state matters.  Second, trajectory monotonicity is not the same as isotonicity on a
complete lattice.  Third, loss of a classical proof invariant is not a complexity theorem.
Within these limits, the local-operator language gives a precise account of the standard
DA proof architecture and a reusable checklist for variants.

\end{document}